\documentclass[journal]{IEEEtran}
\usepackage{stfloats}
\usepackage{float}
\usepackage{graphicx}
\usepackage{subfigure}
\usepackage{color}
\usepackage{algpseudocode}
\usepackage{amssymb}
\usepackage{amsmath}
\usepackage{algorithm}
\usepackage{bm}
\usepackage[colorlinks,citecolor=blue,linkcolor=blue,urlcolor=blue]{hyperref}

\newtheorem{theorem}{Theorem}

\newtheorem{proof}{Proof}

\newcommand{\eig}{\mathrm{Eig}}

\begin{document}

\title{Variable Block-Correlation Modeling\\ and Optimization for Secrecy Analysis in\\Fluid Antenna Systems}

\author{Tuo Wu, 
	         Kwai-Man Luk, \IEEEmembership{Life Fellow,~IEEE},
            Jie Tang, 
            Kai-Kit Wong, \emph{Fellow}, \emph{IEEE}, 
            Jianchao Zheng, 
            Baiyang Liu, \\
            David Morales-Jimenez, 
            Maged Elkashlan, 
            Kin-Fai Tong, \emph{Fellow, IEEE}, 
            Chan-Byoung Chae, \emph{Fellow, IEEE}, 
            Fumiyuki Adachi, \IEEEmembership{Life Fellow,~IEEE}, and 
            George K. Karagiannidis,  \emph{Fellow, IEEE} 
\vspace{-5mm}

\thanks{(\textit{Corresponding author: Kai-Kit Wong.})}

\thanks{This research work of T. Wu was funded by Hong Kong Research Grants Council under the Area of Excellence Scheme under Grant AoE/E-101/23-N. T. Wu and K.-M. Luk  are with the State Key Laboratory of Terahertz and Millimeter Waves, Department of Electronic Engineering, City University of Hong Kong, Hong Kong (E-mail: $\rm \{tuowu2, eekmluk\}@cityu.edu.hk$).
J. Tang is with the School of Electronic and Information Engineering, South China University of Technology, Guangzhou 510640, China (E-mail: $\rm eejtang@scut.edu.cn$).
K. K. Wong is with the Department of Electronic and Electrical Engineering, University College London, WC1E 6BT London, U.K., and also with the Yonsei Frontier Laboratory and the School of Integrated Technology, Yonsei University, Seoul 03722, South Korea (E-mail:$\rm  kai$-$\rm kit.wong@ucl.ac.uk$).
J. Zheng is with the School of Computer Science and Engineering, Huizhou University, Huizhou 516007, China (E-mail: $\rm zhengjch@hzu.edu.cn$).
B. Liu and  K. F. Tong are with the School of Science and Technology, Hong Kong Metropolitan University, Hong Kong SAR, China. (E-mail: $\rm \{byliu,ktong\}@hkmu.edu.hk$). 
D. M. Jimenez is with Department of Signal Theory, Networking and 	Communications, Universidad de Granada, Granada 18071, Spain (E-mail: $\rm dmorales@ugr.es$).
M. Elkashlan is with the School of Electronic Engineering and Computer Science at Queen Mary University of London, London E1 4NS, U.K. (E-mail: $\rm maged.elkashlan@qmul.ac.uk$). 
C.-B. Chae is with the School of Integrated Technology, Yonsei University, Seoul 03722 Korea. (E-mail: $\rm cbchae@yonsei.ac.kr$).
F. Adachi is with the International Research Institute of Disaster Science (IRIDeS), Tohoku University, Sendai, Japan (E-mail: $\rm adachi@ecei.tohoku.ac.jp$).
G. K. Karagiannidis is with the Department of Electrical and Computer Engineering, Aristotle University of Thessaloniki, 54124 Thessaloniki, Greece (E-mail: $\rm geokarag@auth.gr$).} 
}

\maketitle

\begin{abstract}
Fluid antenna systems (FAS) are emerging as a transformative enabler for sixth-generation (6G) wireless communications, providing unprecedented spatial diversity through dynamic reconfiguration of antenna ports. However, the inherent spatial correlation among ports poses significant challenges for accurate analysis. Conventional models such as Jakes are analytically intractable, while oversimplified constant-correlation models fail to capture the true behavior. In this work, we address these challenges by applying the variable block-correlation model (VBCM)---originally proposed by Ram\'{i}rez-Espinosa \textit{et al.} in 2024---to FAS security analysis, and by developing comprehensive optimization methods to enhance analytical accuracy. We derive new closed-form expressions for average secrecy capacity (ASC) and secrecy outage probability (SOP), demonstrating that the VBCM framework achieves simulation-aligned accuracy, with relative errors consistently below $5\%$ (compared to $10$--$15\%$ for constant-correlation models). To maximize ASC, we further design two algorithms: a grid search (GS) method and a gradient descent (GD) method. Numerical results reveal that the VBCM-based approach not only provides reliable insights into FAS security performance, but also yields substantial gains---ASC improvements exceeding $120\%$ in high-threat scenarios and $18$--$19\%$ performance enhancements for compact antenna configurations. \textit{These findings underscore the practical value of integrating VBCM into FAS security analysis and optimization, establishing it as a powerful tool for advancing 6G communication systems.}
\end{abstract}


\section{Introduction}

\IEEEPARstart{T}{he} evolution of wireless communications towards sixth-generation (6G) networks is driving a demand for cutting-edge technologies capable of delivering ultra-high data rates, improved reliability, expanded coverage, and the capacity for massive connectivity \cite{Tariq-2020,Andrews20246G}. These enhancements will facilitate transformative applications, including immersive experiences, ultra-reliable low-latency communications (URLLC), and integrated sensing and communication (ISAC) \cite{ITUwhite}. 

Currently, massive multiple-input multiple-output (MIMO) systems serve as a fundamental component of contemporary wireless networks, utilizing many antennas to exploit spatial diversity \cite{Marzetta-2010}. Nevertheless, the performance of massive MIMO is dictated by the quantity of the available radio-frequency (RF) chains at the terminals.  To achieve the ambitious objectives of 6G, a prominent technology involves augmenting the number of antennas at the base station (BS), leading to the development of extra-large MIMO (XL-MIMO) \cite{10379539}. Nonetheless, rising hardware expenses, increased power consumption, and greater operational complexity pose substantial challenges to the scalability of traditional MIMO systems.

Fluid antenna systems (FAS) \cite{Wong-2020ell,New24,TWu20243}, sometimes referred to as movable antennas if mechanically movable antennas are concerned \cite{Zhu-Wong-2024,LZhu25,LZhu23}, have emerged as a viable solution to address these challenges. First proposed by Wong {\em et al.}~in \cite{FAS20,FAS21}, the emergence of FAS aims to integrate reconfigurable antenna architectures into physical-layer design and network optimization. The concept of FAS treats antenna as a reconfigurable, physical-layer resource, facilitating new coding, signal processing and network optimization schemes, and inspiring new reconfigurable antenna designs. Many efforts have since been made to utilize FAS for enhancing the performance in numerous applications, e.g., \cite{YaoJ241, HXu23, XLai23, BC24, LaiX242, YaoJ251, JYao2024, YaoJ252, Ghadi-2023, New-twc2023, NWaqar23, CWang24}.  

Early prototypes investigated fluidic implementations utilising conductive liquids within precision-controlled structures \cite{shen2024design}. However, technologies such as reconfigurable metasurfaces \cite{BLiu25}, and spatially distributed antenna pixels \cite{Zhang25}, are more effective by eliminating any response time for the FAS reconfigurability. This reconfigurability represents a significant difference from fixed-position antenna (FPA) arrays, in which element locations remain permanently static.  Dynamic control of spatial points as radiating elements in FASs enhances spatial degrees of freedom (DoFs), resulting in a virtual aperture that exceeds the constraints of conventional antenna arrays. The increased flexibility provides notable performance advantages but also raises new concerns regarding system security and reliability in real-world applications. 

FAS fundamentally functions via intelligent port selection mechanisms, enabling the system to dynamically activate and deactivate various antenna ports to enhance communication performance \cite{New24,Wong-frontiers22}. The adaptive capability enables FAS to utilize spatial diversity and mitigate channel fading through the selection of optimal port positions. Nonetheless, the reconfigurability that facilitates performance improvement also introduces significant security concerns that require thorough analysis. Physical-layer security (PLS) has received considerable attention as an effective method for securing wireless communications independent of cryptographic keys \cite{Wyner1975,HNiututorial}. Within the framework of FAS, PLS introduces distinct challenges and novel opportunities, as examined in several recent studies \cite{Shojaeifard,Security1,Security2,Security3,Security4,Security5}.  

Thus, the security analysis of FAS must consider the adaptive port selection strategies of both the legitimate receiver and adversary, and accurately model the complex spatial correlation structure, which presents significant analytical challenges.  The proximity of antenna ports in FAS results in considerable spatial correlation among channel coefficients, necessitating precise characterization to evaluate system performance and security properties.  Accurately modelling the FAS channel is a real challenge, and the channel covariance matrix typically demonstrates a Toeplitz structure, making analytical treatment intractable, unless an overly-simplified model is used \cite{FAS22}.  

To make the analysis tractable, previous studies have proposed spatial block-correlation models that divide the correlated channel coefficients (ports) into independent blocks \cite{BC24}. While this approach is general and tractable, some limitations arise in capturing complex channel structures, especially in the case with fewer ports \cite{BC24, LaiX242,FAS22, Chai22, FAMS}. The spatial block-correlation model, introduced in \cite{BC24}, established a general framework allowing for block-specific correlation coefficients. To illustrate the model's potential, the original work focused on the simplified case of a common correlation parameter. Building upon this, \cite{FAS_BC_paper} later developed a comprehensive algorithm to optimize distinct correlation values for each block, thereby fully enabling the variable block-correlation model (VBCM) and enhancing its accuracy.


The incorporation of VBCM into FAS security analysis presents two significant advantages compared to the constant correlation methods.  VBCM facilitates improved modeling of the intricate spatial correlation structure present in FAS deployments by permitting correlation coefficients to differ across various blocks, thereby accurately reflecting the non-uniform correlation patterns \cite{BC24, FAS_BC_paper}.  The variable correlation framework offers enhanced analytical tractability and high accuracy, resulting in more dependable performance analysis and improved system design optimization \cite{BC24, FAS_BC_paper}.

The introduction of VBCM into FAS security analysis poses notable theoretical and computational challenges. The main challenge is the intricate mathematical framework of VBCM, necessitating advanced optimization algorithms to ascertain optimal block-specific correlation coefficients via eigenvalue-based parameter optimization.  The resulting channel statistics involve complex multi-block probability distributions that cannot be expressed in closed form, requiring advanced numerical integration techniques like Gauss-Chebyshev quadrature for precise evaluation of secrecy metrics. The optimization of FAS security parameters under VBCM is computationally intensive, necessitating repeated evaluations of complex integrals that involve Marcum Q-functions and multi-dimensional probability distributions. This complexity renders real-time optimization difficult for practical applications.

This paper presents a comprehensive VBCM-based security analysis framework designed to systematically address the theoretical and computational complexities associated with these challenges. The VBCM optimization problem is initially formulated to identify optimal block-specific correlation coefficients via eigenvalue-based parameter optimization. We derive the statistical distributions of maximum channel amplitudes under VBCM, involving multi-block probability distributions that necessitate advanced numerical integration techniques. We utilize Gauss-Chebyshev quadrature to assess the complex integrals associated with Marcum Q-functions, which are pertinent to the calculations of secrecy capacity and outage probability. We develop optimization algorithms that effectively balance computational complexity and solution quality, facilitating the practical implementation of VBCM-based FAS security optimization. This work's primary contributions are summarized as follows:
\begin{itemize}
\item \textbf{\textit{Advanced Security Analysis Framework}}---We develop a comprehensive security analysis framework for FAS-enabled secure  communication systems using  VBCM. Unlike conventional approaches that assume uniform constant correlation, our framework is based on the block-correlation model of \cite{BC24}, which allows correlation coefficients to vary across different blocks. We specifically develop a novel method to optimize these correlation parameters, providing superior accuracy (compared to uniform correlation models) especially for FAS deployments with a small number of ports ($N<20$).  
\item \textbf{\textit{Rigorous Mathematical Formulation}}---We derive analytical expressions for the average secrecy capacity (ASC) and secrecy outage probability (SOP). Our analysis provides closed-form expressions that enable efficient performance evaluation while maintaining high accuracy.
\item \textbf{\textit{Comprehensive Optimization Algorithms}}---We develop two optimization algorithms, grid search (GS) and gradient descent (GD) algorithms, to maximize the secrecy performance under various system constraints. These algorithms provide a complete framework for practical FAS security optimization, balancing between global optimality and computational efficiency.
\item \textbf{\textit{Extensive Performance Validation}}---We conduct comprehensive numerical simulations to validate our theoretical framework across various system configurations. The results demonstrate that VBCM-based security evaluation achieves significantly tighter agreement with   simulations compared to constant-correlation approaches.
\item \textbf{\textit{Practical Design Insights}}---Our analysis reveals that FAS can achieve substantial security improvements  with ASC gains exceeding $120\%$ in certain scenarios. We provide practical guidelines for FAS parameter selection and optimization that can inform real-world system deployments.
\end{itemize}

The remainder of this paper is organized as follows. Section \ref{sec:system} introduces the system model for FAS-enabled secure communication and the block-correlation model. Then Section \ref{sec:security} presents the security analysis framework, including the derivation of ASC and SOP expressions. In Section \ref{sec:opt}, we detail the optimization algorithms for FAS security enhancement. Afterwards, Section \ref{sec:analysis} provides extensive numerical results and performance validation. Finally, Section \ref{sec:conclude} concludes the paper and discusses future research directions.

\section{System Model for FAS-Enabled Secure Communication}\label{sec:system}

\begin{figure}[t]
\centering
\includegraphics[width=.9\columnwidth]{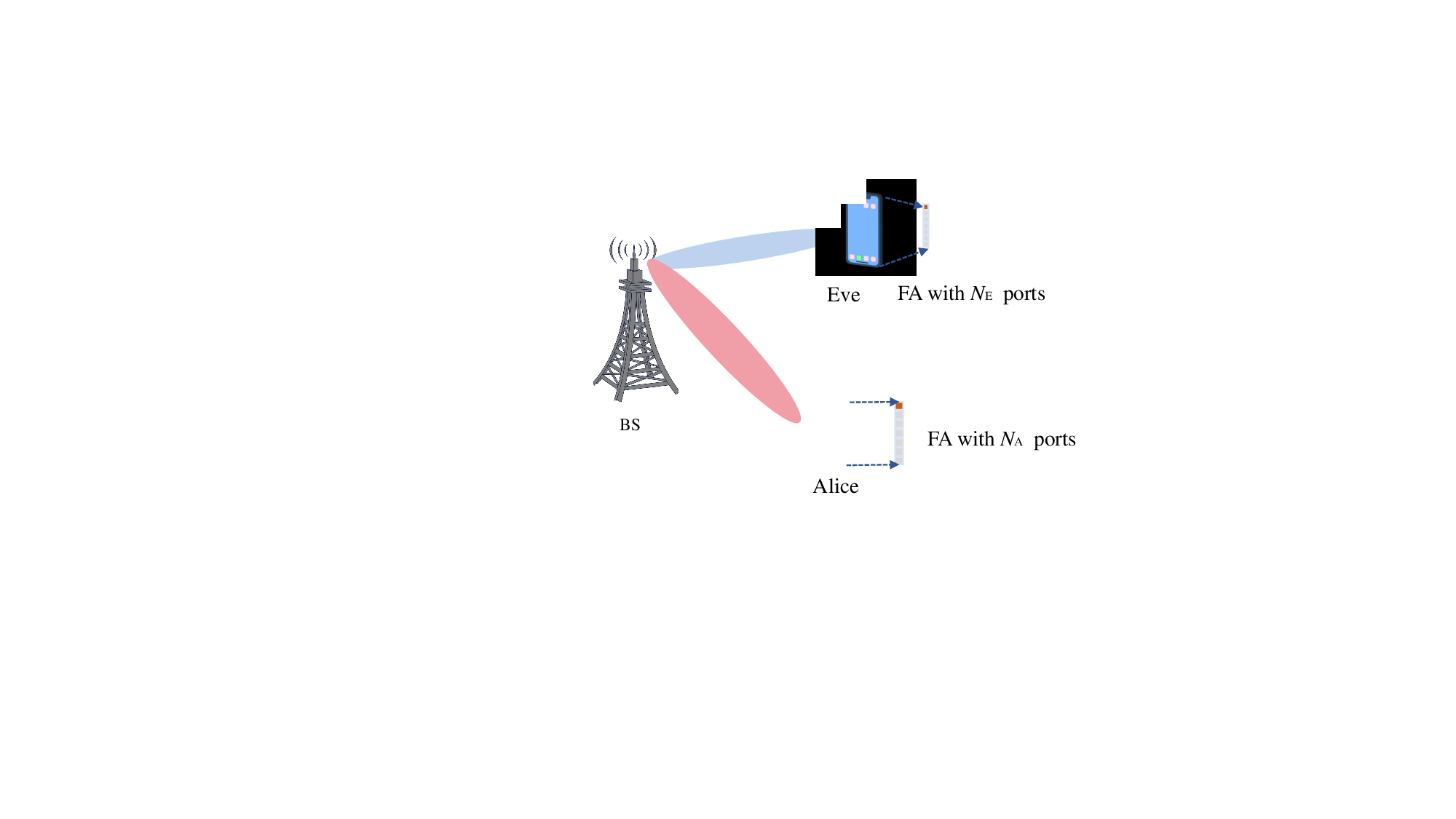}
\caption{The system model of FAS-enabled secure communication.}\label{system_model}
\end{figure}

As illustrated in Fig.~\ref{system_model}, we consider a FAS-enabled secure communication system comprising a single-antenna BS, a legitimate user (Alice), and an eavesdropper (Eve) operating in Rayleigh fading channels. Both Alice and Eve are equipped with one-dimensional linear fluid antennas enabling switching among $N_A$ and $N_E$ discrete ports, respectively. These ports are distributed within a linear space of length $W\lambda$, where $W$ is a normalization parameter and $\lambda$ is the carrier wavelength. Moreover, a rich-scattering model is assumed and the number of scatterers surround Alice and Eve uniformly.

\subsection{Communication Model}
The BS transmits a signal $s$ (with $\mathbf{E}[|s|^2]=1$) with power $\mathcal{P}$. The signal received at the $k$-th port of Alice is given by
\begin{align}
y_{A,k} = \sqrt{\mathcal{P}} g_{A,k} s + n_{A,k}, \quad k \in \{1, \dots, N_A\},
\end{align}
where $g_{A,k} \sim \mathcal{CN}(0, \eta_A)$ is the complex Gaussian channel coefficient from the BS to Alice's $k$-th port, and $n_{A,k} \sim \mathcal{CN}(0, \sigma_A^2)$ represents additive white Gaussian noise. A similar model applies to Eve with parameters $g_{E,k} \sim \mathcal{CN}(0, \eta_E)$ and $n_{E,k} \sim \mathcal{CN}(0, \sigma_E^2)$.

Both users select the port with their highest instantaneous channel gain, which leads to
\begin{equation}
\left\{\begin{aligned}
i^* &= \arg\max_{k \in \{1, \dots, N_A\}} |g_{A,k}|,\\
j^* &= \arg\max_{k \in \{1, \dots, N_E\}} |g_{E,k}|.
\end{aligned}\right.
\end{equation}
For convenience, the resulting maximum channel amplitudes are denoted as $\alpha_A^* = |g_{A,i^*}|$ and $\alpha_E^* = |g_{E,j^*}|$.

\subsection{FAS Channel Correlation Model}
Due to small spacing between adjacent ports, the channel coefficients are spatially correlated. The channel vector $\mathbf{g}_A = [g_{A,1}, \dots, g_{A,N_A}]^T$ has the covariance matrix $\mathbf{\Sigma}_A = \mathbf{E}[\mathbf{g}_A \mathbf{g}_A^H]$. Following the Jakes model, the correlation coefficient between ports $k$ and $l$ is expressed as
\begin{align}
\sigma_{k,l} = J_0\left(\frac{2\pi |k-l| W}{N_A-1}\right),
\end{align}
where $J_0(\cdot)$ is the zero-order Bessel function of the first kind. The covariance matrix is formulated as
\begin{align}
\mathbf{\Sigma}_A = \eta_A \begin{pmatrix}
			\sigma_{1,1} & \sigma_{1,2} & \cdots & \sigma_{1,N_A} \\
			\sigma_{1,2} & \sigma_{1,1} & \cdots & \sigma_{1,N_A-1} \\
			\vdots &  & \ddots & \vdots \\
			\sigma_{1,N_A}& \sigma_{1,N_A-1} & \cdots & \sigma_{1,1}
\end{pmatrix}.
\end{align}
	
\section{Security Performance Analysis}\label{sec:security}
Having established the FAS system model with spatially correlated channels, we now analyze the security performance of FAS-enabled communication \cite{Wyner1975}.
	

\subsection{Secrecy Performance Metrics}
Given the maximum channel amplitudes $\alpha_A^*$ and $\alpha_E^*$ obtained from the port selection process described above, the instantaneous signal-to-noise ratios (SNR) at the selected ports are expressed as
\begin{equation}\label{eq:snr_def}
\left\{\begin{aligned}
\gamma_A &= \frac{\mathcal{P}(\alpha_A^*)^2}{\sigma_A^2},\\
\gamma_E &= \frac{\mathcal{P}(\alpha_E^*)^2}{\sigma_E^2}.
\end{aligned}\right.
\end{equation}
Hence, the achievable data rates for Alice and Eve are respectively formulated as
\begin{equation}\label{eq:capacity_def}
\left\{\begin{aligned}
C_A &= \log_2(1 + \gamma_A),\\
C_E &= \log_2(1 + \gamma_E).
\end{aligned}\right.
\end{equation}
Thus, the secrecy capacity is expressed as  \cite{Wyner1975}
\begin{align}\label{eq:cs_def}
C_s = [C_A - C_E]^+,
\end{align}
where $[x]^+=\max\{0,x\}$.

To comprehensively evaluate the secrecy performance, we consider two fundamental metrics widely adopted in PLS analysis. First of all, the ASC  quantifies the  secrecy performance and is defined as
\begin{align}\label{eq:asc_def}
\bar{C}_s = \mathbf{E}[C_s] = \mathbf{E}[[C_A - C_E]^+].
\end{align} 
Besides, the SOP  characterizes the reliability of secure communication by measuring the probability that the instantaneous secrecy capacity falls below a target secrecy rate $R_s > 0$. Hence, SOP is defined as
\begin{align}\label{eq:sop_def}
P_{\text{sop}} = \Pr(C_s < R_s) = \Pr([C_A - C_E]^+ < R_s).
\end{align}

\subsection{VBCM }
To evaluate the secrecy metrics defined in \eqref{eq:asc_def} and \eqref{eq:sop_def}, we need the statistical distributions of the maximum channel amplitudes $\alpha_A^*$ and $\alpha_E^*$. However, the spatial correlation structure described by the Toeplitz covariance matrix $\mathbf{\Sigma}_h$ in Section \ref{sec:system} renders direct analytical evaluation intractable.

To overcome this challenge, we employ the  VBCM as a mathematically tractable approximation \cite{BC24}. For an $N \times N$ Toeplitz matrix $\mathbf{\Sigma}_h$, this model approximates it with a block-diagonal structure that captures the essential correlation characteristics while enabling analytical progress.

The VBCM partitions the $N$ correlated channels into $D$ independent blocks of sizes $\{L_d\}_{d=1}^D$ such that the constraint $\sum_{d=1}^D L_d = N$ is satisfied. Within each block $d$, the correlation structure is simplified to a constant correlation model, where the correlation matrix is expressed as \cite{BC24}
\begin{align}\label{eq:block_corr_matrix}
\mathbf{A}_d = \begin{pmatrix}
			1 & \rho_d & \cdots & \rho_d \\
			\rho_d & 1 & \cdots & \rho_d \\
			\vdots &  & \ddots & \vdots \\
			\rho_d & \rho_d& \cdots & 1 
\end{pmatrix} \in\mathbb{R}^{L_d\times L_d},
\end{align}
where $\rho_d \in [0, 1]$ is a variable correlation coefficient for block $d$. To ensure the VBCM provides an accurate approximation of the original Toeplitz matrix, the parameters $D$, $\{L_d\}$, and $\{\rho_d\}$ are optimized by minimizing the eigenvalue distance between the original and approximated matrices. This optimization problem is formulated as
\begin{align}\label{eq:obj_func}
\underset{D, \{\rho_d\}, \{L_d\}}{\min} \|\eig(\mathbf{\Sigma}_h) - \eig(\mathbf{\hat{\Sigma}}_h)\|_2^2,
\end{align}
where $\mathbf{\hat{\Sigma}}_h$ represents the block-diagonal approximation of $\mathbf{\Sigma}_h$.

The optimization can be decomposed into block-wise subproblems. For each block $d$, the assignment error when eigenvalue $\lambda_{\text{cur}}$ is added to block $d$ is quantified by
\begin{multline}\label{eq:block_error}
\mathrm{dist}_d = (1 + \rho_d(L_d-1) - \lambda_d)^2\\
 + \sum_{k \in K_d \cup \{\lambda_{\text{cur}}\}} (\lambda_k - 1 + \rho_d)^2,
\end{multline}
where $K_d$ contains the eigenvalues currently assigned to block $d$, and $L_d = |K_d| + 1$ is the resulting block size.

For a given block assignment, the optimal correlation coefficient $\rho_d^*$ that minimizes the block error is obtained by solving the least-squares problem:
\begin{align}\label{eq:optimal_rho}
\rho_d^* = \min\left(\max\left(\frac{(L_d-1)\lambda_d - \sum_{k \in K_d \cup \{\lambda_{\text{cur}}\}} \lambda_k}{2(L_d-1)}, 0\right), 1\right).
\end{align}

\begin{algorithm}[h!] \label{algorithm1}
\caption{Variable Block-Correlation Parameter Optimization}
\begin{algorithmic}[1]
\Require $D$ (number of blocks), $\eig(\mathbf{\Sigma}_h)=[\lambda_1, \dots, \lambda_N]$ (eigenvalues in descending order)
\State Split eigenvalues: $\Lambda_1=[\lambda_1,\dots,\lambda_D]$ (dominant eigenvalues), $\Lambda_2=[\lambda_{D+1},\dots,\lambda_N]$ (remaining eigenvalues)
\State Initialize empty block eigenvalue sets $K_d=[\,]$ and $\rho_d^*=1$ for $d=1,\dots,D$
			\For{$n=1$ to $N-D$} \Comment{Distribute remaining eigenvalues to blocks}
			\State Let $\lambda_{\text{cur}} = \lambda_{D+n}$ \Comment{Current eigenvalue to assign}
			\For{$d=1$ to $D$} \Comment{Evaluate assignment to each block}
			\State Compute $L_d = |K_d| + 1$ (block size if $\lambda_{\text{cur}}$ is assigned to block $d$)
			\State Calculate optimal $\tilde{\rho}_d$ using \eqref{eq:optimal_rho}
			\State Compute assignment error using \eqref{eq:block_error}
			\EndFor
			\State Find best assignment: $d^* = \arg\min_{d \in \{1,\dots,D\}} \mathrm{dist}_d$
			\State Update block: $K_{d^*} \leftarrow [K_{d^*}, \lambda_{\text{cur}}]$ and $\rho_{d^*}^* \leftarrow \tilde{\rho}_{d^*}$
			\EndFor
			\State Determine final block sizes: $L_d = |K_d| + 1$ for $d=1,\dots,D$
			\Ensure Optimal parameters $\{L_d\}_{d=1}^D$ and $\{\rho_d^*\}_{d=1}^D$
		\end{algorithmic}
\end{algorithm}

The algorithm operates on the principle that the eigenvalue structure of the block-diagonal approximation should match that of the original Toeplitz matrix as closely as possible. The first $D$ largest eigenvalues ($\Lambda_1$) are reserved as the dominant eigenvalues for each block, while the remaining smaller eigenvalues ($\Lambda_2$) are distributed among the $D$ blocks to minimize the overall approximation error. The algorithm uses a greedy assignment strategy: for each remaining eigenvalue $\lambda_{\text{cur}}$, it evaluates the cost of assigning it to each block by computing the resulting optimal correlation coefficient $\tilde{\rho}_d$ and the corresponding assignment error $\mathrm{dist}_d$. The eigenvalue is then assigned to the block that yields the minimum error. This approach ensures that eigenvalues with similar magnitudes are grouped together, leading to more accurate block-wise approximations while maintaining computational efficiency with only $(N-D) \times D$ evaluations instead of the exponential complexity $D^{N-D}$ required by exhaustive search.

\subsection{Channel Statistics Derivation}
With the VBCM framework established through the optimization problem \eqref{eq:obj_func}, we are ready to derive the statistical distributions of the maximum channel amplitudes required for evaluating the secrecy metrics in \eqref{eq:asc_def} and \eqref{eq:sop_def}. This derivation follows the established methodology for analyzing correlated channel statistics in FASs such as in \cite{FAS21, FAS22}.

Consider a generic FAS user with $N$ ports and covariance matrix $\mathbf{\Sigma}$. Applying Algorithm \textcolor{blue}{1} yields optimal VBCM parameters $(D, \{L_d\}_{d=1}^D, \{\rho_d\}_{d=1}^D)$, which transforms the original correlated channel vector into $D$ independent blocks as specified by the block structure in \eqref{eq:block_corr_matrix}. This transformation enables tractable analysis while maintaining high accuracy \cite{BC24,   FAS_BC_paper}.

Within each block $d$, the $L_d$ correlated channel coefficients can be decomposed using the canonical representation of the constant correlation model, which yields
\begin{align}\label{eq:channel_decomp}
\hat{g}_{k,d} = \sqrt{\eta(1-\rho_d)} p_{k,d} + \sqrt{\eta\rho_d} q_d,
\end{align}
where $k =\{1,\dots,N_A\}$ and $p_{k,d} \sim \mathcal{CN}(0,1)$ are independent and identically distributed (i.i.d.) private variables, and $q_d \sim \mathcal{CN}(0,1)$ is the common variable shared by all the channels in block $d$.

To analyze the distribution of channel amplitudes, we define  $\phi_d = \sqrt{\eta\rho_d}|q_d|$. Given this conditioning variable, $\alpha_{k,d} = |\hat{g}_{k,d}|$ follows a Rician distribution due to the decomposition in \eqref{eq:channel_decomp}. This approach is consistent with the established analysis of correlated Rayleigh fading channels \cite{book}. The conditional cumulative density function (CDF) is given by
\begin{align}\label{eq:rice_cdf}
F_{\alpha_{k,d}|\phi_d}(x|\theta) = 1 - Q_1\left(\frac{\theta}{\sigma_d}, \frac{x}{\sigma_d}\right),
\end{align}
where $\sigma_d^2 = \eta(1-\rho_d)/2$ represents the variance of the scattered component, and $Q_1(\cdot,\cdot)$ is the first-order Marcum Q-function \cite{book}.

Since FAS employs port selection based on the maximum channel gain, we need the distribution of the maximum amplitude within each block. Using the independence of channels within a block conditioned on $\phi_d$, the conditional CDF of the maximum amplitude within block $d$ is expressed as
\begin{align}\label{eq:block_max_cdf_cond}
F_{\alpha_d^*|\phi_d}(x|\theta) = \left[ F_{\alpha_{k,d}|\phi_d}(x|\theta) \right]^{L_d},
\end{align}
which follows from the order statistics of $L_d$ independent Rician random variables.

To obtain the unconditional distribution, we integrate over the distribution of $\phi_d$. Since $\phi_d$ follows a Rayleigh distribution with probability density function (PDF) $f_{\phi_d}(\theta) = \frac{2\theta}{\eta\rho_d} e^{-\theta^2/(\eta\rho_d)}$, the unconditional CDF is given by
\begin{align}\label{eq:block_cdf}
F_{\alpha_d^*}(x) = \int_0^\infty \left[1 - Q_1\left(\frac{\theta}{\sigma_d}, \frac{x}{\sigma_d}\right)\right]^{L_d} \frac{2\theta}{\eta\rho_d} e^{-\frac{\theta^2}{\eta\rho_d}} d\theta.
\end{align}

Since the $D$ blocks are independent under the VBCM framework, the overall CDF of the user's maximum channel amplitude is obtained as the product of individual block CDFs, which gives
\begin{align}\label{eq:cdf_final}
F_{\alpha^*}(x) = \prod_{d=1}^{D} F_{\alpha_d^*}(x).
\end{align}

The corresponding PDF required for the secrecy capacity analysis is obtained by differentiating \eqref{eq:cdf_final} and expressed as
\begin{align}\label{eq:pdf_final}
f_{\alpha^*}(x) = \sum_{d=1}^{D} \left[ f_{\alpha_d^*}(x) \prod_{j=1\atop j \neq d}^{D} F_{\alpha_j^*}(x) \right],
\end{align}
where $f_{\alpha_d^*}(x) = \frac{d}{dx} F_{\alpha_d^*}(x)$ represents the PDF of the maximum amplitude in block $d$.

\subsection{Derivation of ASC}
We now proceed to evaluate the ASC defined in \eqref{eq:asc_def}. This derivation addresses previous theoretical inconsistencies and provides a rigorous mathematical framework.

Starting from the definition of ASC in \eqref{eq:asc_def}, we expand the expectation operation to obtain
\begin{align}\label{eq:asc_expansion}
\bar{C}_s = \int_0^\infty \int_0^\infty [C_A(x) - C_E(y)]^+ f_{\alpha_A^*}(x) f_{\alpha_E^*}(y) \, dy \, dx,
\end{align}
where the PDFs $f_{\alpha_A^*}(x)$ and $f_{\alpha_E^*}(y)$ are given by \eqref{eq:pdf_final} for Alice and Eve, respectively. Substituting the SNR definitions from \eqref{eq:snr_def} and the capacity expressions from \eqref{eq:capacity_def}, the capacity functions become
\begin{equation}\label{eq:capacity_functions}
\left\{\begin{aligned}
C_A(x) &= \log_2\left(1 + \frac{\mathcal{P}x^2}{\sigma_A^2}\right),\\
C_E(y) &= \log_2\left(1 + \frac{\mathcal{P}y^2}{\sigma_E^2}\right).
\end{aligned}\right.
\end{equation}
Furthermore, the constraint $C_A(x) \geq C_E(y)$ for positive secrecy capacity, derived from the definition in \eqref{eq:cs_def}, yields the integration boundary condition $y \leq \frac{\sigma_E}{\sigma_A}x$.

To evaluate the double integral in \eqref{eq:asc_expansion} with the integration constraint derived above, we decompose the ASC into two manageable components. Using the integration boundary condition, the ASC is decomposed into
\begin{align}\label{eq:asc_decomp}
\bar{C}_s &= C_s^{(1)} - C_s^{(2)},
\end{align}
where the two components are formulated as
\begin{align}
C_s^{(1)} &= \int_0^\infty C_A(x) f_{\alpha_A^*}(x) F_{\alpha_E^*}\left(\frac{\sigma_E}{\sigma_A}x\right) dx, \label{eq:cs1}\\
C_s^{(2)} &= \int_0^\infty \int_0^{\frac{\sigma_E}{\sigma_A}x} C_E(y) f_{\alpha_A^*}(x) f_{\alpha_E^*}(y) \, dy \, dx. \label{eq:cs2}
\end{align}
However, the integrals in \eqref{eq:cs1} and \eqref{eq:cs2} cannot be evaluated in closed form due to the complex expressions of the PDFs and CDFs derived in the previous section. Therefore, we employ Gauss-Chebyshev quadrature for numerical evaluation.

For the first integral in \eqref{eq:cs1}, we apply the substitution $x = \frac{H}{2}(t+1)$ with $t \in [-1,1]$ to transform the integration domain. The Jacobian of this transformation is $\frac{dx}{dt} = \frac{H}{2}$. Applying Gauss-Chebyshev quadrature following the numerical integration techniques established in \cite{NumericalAnalysis}, we obtain
\begin{align}\label{eq:cs1_approx}
C_s^{(1)} &\approx \frac{H\pi}{2U_p} \sum_{p=1}^{U_p} \sqrt{1-t_p^2} C_A(\beta_p) f_{\alpha_A^*}(\beta_p) F_{\alpha_E^*}\left(\frac{\sigma_E}{\sigma_A}\beta_p\right),
\end{align}
where $\beta_p = \frac{H}{2}(t_p+1)$ and $t_p = \cos\left(\frac{(2p-1)\pi}{2U_p}\right)$ are the Gauss-Chebyshev nodes.

For the second integral in \eqref{eq:cs2}, we employ a double substitution approach. First, we substitute $x = \frac{H}{2}(t+1)$ for the outer integral. For the inner integral, we substitute $y = \frac{\sigma_E}{\sigma_A}x \cdot \frac{s+1}{2}$ with $s \in [-1,1]$, which gives the constraint-respecting transformation. The combined Jacobian is $J = \frac{H}{2} \cdot \frac{\sigma_E}{\sigma_A} \cdot \frac{x}{2} = \frac{H\sigma_E x}{4\sigma_A}$. Substituting $x = \beta_p$ and applying the double Gauss-Chebyshev quadrature formula yields
\begin{multline}\label{eq:cs2_approx}
C_s^{(2)} \approx \frac{H\pi^2\sigma_E}{8U_p U_l \sigma_A} \sum_{p=1}^{U_p} \sum_{l=1}^{U_l} \sqrt{1-t_p^2}\sqrt{1-q_l^2}\\
\times C_E(\chi_{p,l}) f_{\alpha_A^*}(\beta_p) f_{\alpha_E^*}(\chi_{p,l}),
\end{multline}
where $q_l = \cos\left(\frac{(2l-1)\pi}{2U_l}\right)$ and $\chi_{p,l} = \frac{\sigma_E}{\sigma_A}\beta_p\frac{q_l+1}{2}$ represents the transformed integration variable.

Combining the approximations in \eqref{eq:cs1_approx} and \eqref{eq:cs2_approx} with the decomposition in \eqref{eq:asc_decomp}, and substituting the capacity functions from \eqref{eq:capacity_functions}, the complete expression for ASC is written as
\begin{align}\label{eq:asc_rigorous}
\bar{C}_s &\approx \frac{H\pi}{2U_p} \sum_{p=1}^{U_p} \sqrt{1-t_p^2} \log_2\left(1+\frac{\mathcal{P}\beta_p^2}{\sigma_A^2}\right) \notag \\
		&\quad \times f_{\alpha_A^*}(\beta_p) F_{\alpha_E^*}\left(\frac{\sigma_E}{\sigma_A}\beta_p\right) \notag \\
		&\quad - \frac{H\pi^2\sigma_E}{8U_p U_l \sigma_A} \sum_{p=1}^{U_p} \sum_{l=1}^{U_l} \sqrt{1-t_p^2}\sqrt{1-q_l^2}   \notag \\
		&\quad \times \log_2\left(1+\frac{\mathcal{P}\chi_{p,l}^2}{\sigma_E^2}\right) f_{\alpha_A^*}(\beta_p) f_{\alpha_E^*}(\chi_{p,l}),
\end{align}
where the quadrature parameters are defined as
\begin{align}\label{eq:quadrature_params}
t_p &= \cos\left(\frac{(2p-1)\pi}{2U_p}\right), \quad \beta_p = \frac{H}{2}(t_p+1), \\
q_l &= \cos\left(\frac{(2l-1)\pi}{2U_l}\right), \quad \chi_{p,l} = \frac{\sigma_E}{\sigma_A}\beta_p\frac{q_l+1}{2}.
\end{align}

\subsection{Derivation of SOP}
Using the definition in \eqref{eq:sop_def} and the secrecy capacity expression from \eqref{eq:cs_def}, the SOP can be expressed as
\begin{align}\label{eq:sop_integral}
P_{\text{sop}} = \int_0^\infty F_{\alpha_A^*}\left(\sqrt{\frac{\sigma_A^2}{\mathcal{P}}\left(2^{R_s}\left(1+\frac{\mathcal{P}y^2}{\sigma_E^2}\right)-1\right)}\right) f_{\alpha_E^*}(y) dy,
\end{align}
where the integration is performed over Eve's channel amplitude distribution, and the CDF of Alice's channel amplitude is evaluated at the threshold value required for achieving the target secrecy rate. Similar to the ASC derivation, this integral cannot be evaluated in closed-form due to the complex expressions of the CDFs and PDFs derived in \eqref{eq:cdf_final} and \eqref{eq:pdf_final}. Using Gauss-Chebyshev quadrature with the same transformation approach, the approximation is given by
\begin{multline}\label{eq:sop_approx}
P_{\text{sop}} \approx \frac{H\pi}{2U} \sum_{p=1}^{U} \sqrt{1-t_p^2} f_{\alpha_E^*}(\beta_p)\\
\times F_{\alpha_A^*}\left( \sqrt{ \frac{\sigma_A^2}{\mathcal{P}} \left( 2^{R_s}\left(1+\frac{\mathcal{P}\beta_p^2}{\sigma_E^2}\right) - 1 \right) } \right).
\end{multline}

\section{Optimization Problem Formulation and Algorithm Design}\label{sec:opt}
In the preceding section, we developed a rigorous analytical framework to evaluate key security metrics, such as ASC and SOP. This framework provides the tools to accurately predict security system performance under various configurations. Building upon this foundation, this section aims to systematically optimize the FAS parameters for maximizing security performance. We formulate the problem with the objective of maximizing ASC under practical constraints and present two distinct algorithmic solutions, GS and GD, to find the optimal operating points for the FAS.

\subsection{FAS Security Optimization Problem}
The objective in FAS security analysis is to maximize the ASC while considering practical constraints such as power limitations and hardware complexity. This leads to the multi-objective optimization problem:
\begin{align}\label{eq:optimization_problem}
		\max_{N_A, P_{tx}} \quad &\bar{C}_s(N_A, P_{tx}) \\
		\text{s.t.  } \quad &P_{tx} \leq P_{\max}, \nonumber \\
		&N_A \leq N_{\max}, \nonumber \\
		&\bar{C}_s(N_A, P_{tx}) \geq 0, \nonumber
\end{align}
where $\bar{C}_s(N_A, P_{tx})$ is the ASC as a function of Alice's port number $N_A$ and transmit power $P_{tx}$, derived in \eqref{eq:asc_rigorous}.

The optimization problem in \eqref{eq:optimization_problem} is inherently complex due to the non-convex nature of the objective function and the discrete nature of the port selection variable $N_A$. Moreover, the presence of the eavesdropper with similar FAS capabilities creates a competitive scenario where the optimization landscape exhibits multiple local optima.

To address this challenging optimization problem, we propose two distinct algorithmic approaches, each with unique advantages and computational characteristics.

\subsection{GS Optimization}
The GS algorithm provides a systematic approach to explore the entire feasible parameter space, guaranteeing global optimality at the cost of computational complexity.

The GS method discretizes  Problem \eqref{eq:optimization_problem} into a finite set of candidate solutions. For the power parameter $P_{tx}$, we construct a uniform grid, which is written as
\begin{align}\label{eq:power_grid}
\mathcal{P} = \left\{P_{tx}^{\min} + k \cdot \frac{P_{tx}^{\max} - P_{tx}^{\min}}{G-1} : k = 0, 1, \ldots, G-1\right\},
\end{align}
where $G$ is the grid resolution. The grid spacing is given by 
\begin{align}\label{eq:grid_spacing}
\Delta P = \frac{P_{tx}^{\max} - P_{tx}^{\min}}{G-1}.
\end{align} 
For the discrete port parameter $N_A$, the search space is naturally defined as 
\begin{align}\label{eq:port_grid}
\mathcal{N} = \{N_A^{\min}, N_A^{\min}+1, \ldots, N_A^{\max}\}.
\end{align}
Thus, the total number of candidate solutions is given by
\begin{align}\label{eq:total_candidates}
|\mathcal{S}| = |\mathcal{P}| \times |\mathcal{N}| = G \times (N_A^{\max} - N_A^{\min} + 1).
\end{align}
For each candidate solution $(N_A, P_{tx}) \in \mathcal{N} \times \mathcal{P}$, the ASC is computed with \eqref{eq:asc_rigorous}. The GS optimization seeks the global maximum results, which is expressed as
\begin{align}\label{eq:grid_optimization}
(N_A^*, P_{tx}^*) = \arg\max_{(N_A, P_{tx}) \in \mathcal{N} \times \mathcal{P}} \bar{C}_s(N_A, P_{tx}).
\end{align}
Hence, the optimization process is formulated as
\begin{align}
\bar{C}_s^* &= \max_{N_A \in \mathcal{N}} \max_{P_{tx} \in \mathcal{P}} \bar{C}_s(N_A, P_{tx})\notag \\
&= \max_{i=1,\ldots,|\mathcal{N}|} \max_{j=1,\ldots,G} \bar{C}_s(N_A^{(i)}, P_{tx}^{(j)}), \label{eq:grid_double_max}
\end{align}
where $N_A^{(i)}$ and $P_{tx}^{(j)}$ represent the $i$-th and $j$-th grid points respectively. 

The GS algorithm achieves global optimality in the discrete sense. As the grid resolution increases, the approximation error decreases according to:
\begin{align}\label{eq:grid_error}
\epsilon_{grid} = \left|\bar{C}_s^{true} - \bar{C}_s^*\right| \leq L \cdot \frac{\Delta P}{2},
\end{align}
where $L$ is the Lipschitz constant of $\bar{C}_s$ with respect to $P_{tx}$, and $\bar{C}_s^{true}$ is the true continuous optimum.

The GS algorithm has computational complexity $\mathcal{O}(G \cdot (N_A^{\max} - N_A^{\min} + 1))$, where each objective function evaluation requires numerical integration with complexity $\mathcal{O}(U_p \cdot U_l)$. The total computational cost is thus given by
\begin{align}\label{eq:grid_complexity}
\mathcal{C}_{grid} = G \cdot (N_A^{\max} - N_A^{\min} + 1) \cdot U_p \cdot U_l \cdot \mathcal{C}_{eval},
\end{align}
where $\mathcal{C}_{eval}$ represents the cost of computing the VBCM channel statistics.

\subsection{GD Optimization}
The GD algorithm offers rapid convergence to local optima through iterative parameter updates based on gradient information. This approach is particularly suitable for real-time optimization scenarios where efficiency is priority \cite{NumericalAnalysis}.

The GD method transforms Problem \eqref{eq:optimization_problem} into a continuous optimization problem for the power parameter $P_{tx}$, while handling the discrete port parameter $N_A$ through periodic exhaustive search. This hybrid approach balances optimization accuracy with computational efficiency \cite{book}.

For the continuous parameter $P_{tx}$, the GD update rule is given by
\begin{align}\label{eq:gradient_update}
P_{tx}^{(t+1)} = P_{tx}^{(t)} + \alpha \cdot \nabla_{P_{tx}} \bar{C}_s(N_A^{(t)}, P_{tx}^{(t)}),
\end{align}
where $\alpha > 0$ is the learning rate and $\nabla_{P_{tx}}$ denotes the partial derivative with respect to $P_{tx}$.

The analytical gradient of the ASC with respect to $P_{tx}$ is derived as
\begin{equation}\label{eq:gradient_decomposition}
\frac{\partial \bar{C}_s}{\partial P_{tx}} = \frac{\partial}{\partial P_{tx}} \left[ C_s^{(1)} - C_s^{(2)} \right]
= \frac{\partial C_s^{(1)}}{\partial P_{tx}} - \frac{\partial C_s^{(2)}}{\partial P_{tx}}, 
\end{equation}
where $C_s^{(1)}$ and $C_s^{(2)}$ are defined in \eqref{eq:cs1} and \eqref{eq:cs2}. The first component gradient is written as
\begin{multline}\label{eq:gradient_cs1}
\frac{\partial C_s^{(1)}}{\partial P_{tx}} = \frac{H\pi}{2U_p} \sum_{p=1}^{U_p} \sqrt{1-t_p^2} \frac{\beta_p^2}{\sigma_A^2} \frac{1}{(1+\frac{P_{tx}\beta_p^2}{\sigma_A^2})\ln(2)}\\
\times f_{\alpha_A^*}(\beta_p) F_{\alpha_E^*}\left(\frac{\sigma_E}{\sigma_A}\beta_p\right).
\end{multline} 
The second component gradient is given by
\begin{align}\label{eq:gradient_cs2}
\frac{\partial C_s^{(2)}}{\partial P_{tx}} &= \frac{H\pi^2\sigma_E}{8U_p U_l \sigma_A} \sum_{p=1}^{U_p} \sum_{l=1}^{U_l} \sqrt{1-t_p^2}\sqrt{1-q_l^2} \notag \\
&\quad \times \frac{\chi_{p,l}^2}{\sigma_E^2} \frac{1}{(1+\frac{P_{tx}\chi_{p,l}^2}{\sigma_E^2})\ln(2)} \notag \\
&\quad \times f_{\alpha_A^*}(\beta_p) f_{\alpha_E^*}(\chi_{p,l}).
\end{align}
Due to the complexity of the analytical gradient, we employ finite difference approximation,  which is given by
\begin{align}\label{eq:numerical_gradient}
g_{P_{tx}}^{(t)} = \frac{\bar{C}_s(N_A^{(t)}, P_{tx}^{(t)} + \epsilon) - \bar{C}_s(N_A^{(t)}, P_{tx}^{(t)} - \epsilon)}{2\epsilon},
\end{align}
where $\epsilon$ is a small perturbation parameter. The approximation error is bounded by 
\begin{align}\label{eq:gradient_error}
\left|g_{P_{tx}}^{(t)} - \frac{\partial \bar{C}_s}{\partial P_{tx}}\right| \leq \frac{M \epsilon^2}{6},
\end{align}
where $M$ is an upper bound on the third derivative of $\bar{C}_s$ with respect to $P_{tx}$.

To enforce the power constraint $P_{tx} \in [P_{tx}^{\min}, P_{tx}^{\max}]$, we use the projection operator:
\begin{align}\label{eq:projection_operator}
\text{clip}(x, a, b) = \begin{cases}
a & \text{if } x < a,\\
x & \text{if } a \leq x \leq b,\\
b & \text{if } x > b.
\end{cases}
\end{align}

The constrained update rule becomes
\begin{align}\label{eq:constrained_update}
P_{tx}^{(t+1)} = \text{clip}(P_{tx}^{(t)} + \alpha \cdot g_{P_{tx}}^{(t)}, P_{tx}^{\min}, P_{tx}^{\max}).
\end{align}

For  $N_A$, we perform periodic local search. At iteration $t$, if $t \bmod \Delta t = 0$ (where $\Delta t$ is the search interval), we evaluate 
\begin{align}\label{eq:discrete_search}
\bar{C}_s^{+} = \bar{C}_s(N_A^{(t)} + 1, P_{tx}^{(t+1)}),
\end{align}
and update according to 
\begin{align}\label{eq:discrete_update}
N_A^{(t+1)} = \begin{cases}
N_A^{(t)} + 1 &  \bar{C}_s^{+} > \bar{C}_s(N_A^{(t)}, P_{tx}^{(t+1)})~\mbox{and } N_A^{(t)} < N_A^{\max},\\
N_A^{(t)} & \text{otherwise}.
\end{cases}
\end{align}

The convergence of the GD algorithm depends on the choice of learning rate $\alpha$. For a function with Lipschitz continuous gradient with constant $L$, the convergence rate is given by
\begin{align}\label{eq:convergence_rate}
\bar{C}_s^{(t)} - \bar{C}_s^* \leq \frac{2L\|\mathbf{x}^{(0)} - \mathbf{x}^*\|^2}{t + 4},
\end{align}
where $\mathbf{x}^{(0)}$ is the initial point and $\mathbf{x}^*$ is the optimal solution. The optimal learning rate is given by 
\begin{align}\label{eq:optimal_learning_rate}
\alpha^* = \frac{1}{L},
\end{align}
where $L$ can be estimated from the Hessian of the objective function.

The GD algorithm has computational complexity $\mathcal{O}(T)$, where $T$ is the number of iterations. Each iteration requires $2-3$ function evaluations for gradient computation, making the total cost as
\begin{align}\label{eq:gradient_complexity}
\mathcal{C}_{gradient} = T \cdot \left(2 + \frac{1}{\Delta t}\right) \cdot U_p \cdot U_l \cdot \mathcal{C}_{eval},
\end{align}
where $\frac{1}{\Delta t}$ accounts for the periodic discrete optimization.

\begin{figure*}[t]
\centering
\includegraphics[width=.95\linewidth]{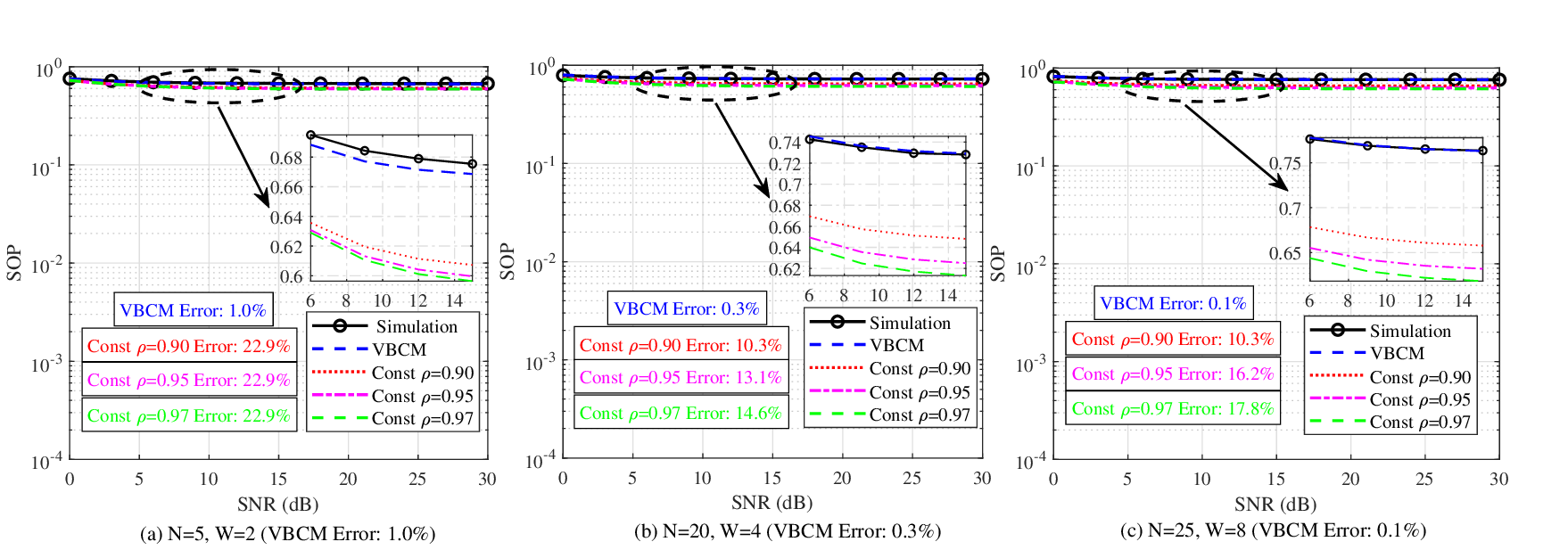}
\caption{SOP comparison using VBCM and uniform constant correlation \cite{FAS22} models: Simulation and theoretical results.}\label{SOPv1}
\end{figure*}

\subsection{Theoretical Convergence Analysis}

\begin{theorem}[GS Optimality] 
The GS algorithm converges to the global optimum of problem \eqref{eq:optimization_problem} with probability 1 as the grid resolution $G \to \infty$.
\end{theorem}

\begin{proof}
Let $\mathcal{S}^* = \arg\max_{(N_A, P_{tx}) \in \mathcal{X}} \bar{C}_s(N_A, P_{tx})$ be the set of global optima, where $\mathcal{X}$ is the feasible region. For any $\epsilon > 0$, there exists a grid resolution $G_\epsilon$ such that for all $G \geq G_\epsilon$, the grid spacing satisfies
\begin{align}
\Delta P = \frac{P_{tx}^{\max} - P_{tx}^{\min}}{G-1} \leq \frac{\epsilon}{L},
\end{align}
where $L$ is the Lipschitz constant of $\bar{C}_s$ with respect to $P_{tx}$.

For any global optimum $(N_A^*, P_{tx}^*) \in \mathcal{S}^*$, there exists a grid point $(N_A^*, P_{tx}^{(j)})$ such that $|P_{tx}^* - P_{tx}^{(j)}| \leq \Delta P/2$. By the Lipschitz continuity:
\begin{align}
|\bar{C}_s(N_A^*, P_{tx}^*) - \bar{C}_s(N_A^*, P_{tx}^{(j)})| \leq L \cdot \frac{\Delta P}{2} \leq \frac{\epsilon}{2}.
\end{align}

Since the GS evaluates all grid points, it finds the maximum among them, which approaches the global maximum as $G \to \infty$. Therefore, we have
\begin{align}
\lim_{G \to \infty} \bar{C}_s^{grid} = \max_{(N_A, P_{tx}) \in \mathcal{X}} \bar{C}_s(N_A, P_{tx}).
\end{align}
\end{proof}

\begin{theorem}[GD Convergence]
Under the assumption that $\bar{C}_s(N_A, P_{tx})$ is locally Lipschitz continuous with respect to $P_{tx}$ with constant $L$, the GD algorithm converges to a stationary point with learning rate $\alpha \in (0, 2/L)$.
\end{theorem}

\begin{proof}
Consider the continuous optimization problem with fixed $N_A$. The GD update can be written as
\begin{align}
P_{tx}^{(t+1)} = P_{tx}^{(t)} + \alpha \nabla_{P_{tx}} \bar{C}_s(N_A, P_{tx}^{(t)}).
\end{align}

For the projected GD with constraints, we have
\begin{align}
P_{tx}^{(t+1)} = \Pi_{[P_{tx}^{\min}, P_{tx}^{\max}]}(P_{tx}^{(t)} + \alpha \nabla_{P_{tx}} \bar{C}_s(N_A, P_{tx}^{(t)})),
\end{align}
where $\Pi_{\mathcal{C}}(\cdot)$ denotes the projection onto the constraint set $\mathcal{C}$.

Under the Lipschitz assumption, the function satisfies
\begin{align}
&\bar{C}_s(N_A, P_{tx}^{(t+1)}) \geq \bar{C}_s(N_A, P_{tx}^{(t)})\nonumber\\
& + \alpha \|\nabla_{P_{tx}} \bar{C}_s(N_A, P_{tx}^{(t)})\|^2 - \frac{\alpha^2 L}{2} \|\nabla_{P_{tx}} \bar{C}_s(N_A, P_{tx}^{(t)})\|^2.
\end{align}

For $\alpha < 2/L$, the coefficient of $\|\nabla_{P_{tx}} \bar{C}_s(N_A, P_{tx}^{(t)})\|^2$ is positive, ensuring convergence to a stationary point where $\nabla_{P_{tx}} \bar{C}_s(N_A, P_{tx}^*) = 0$.

The discrete optimization of $N_A$ is performed through local search, which ensures that $N_A^{(t+1)}$ is chosen to maximize $\bar{C}_s(N_A, P_{tx}^{(t+1)})$ among neighboring values, providing a discrete analog of the optimality condition.
\end{proof}
	
\section{Numerical Results and Analysis}\label{sec:analysis}
This section presents comprehensive numerical results based on multiple simulation scenarios with varying system configurations. The analysis encompasses both theoretical validation and optimization algorithm performance evaluation across different threat levels and system parameters.

\subsection{System Parameters}
The system operates at  $f_c = 2.4$ GHz  with  $\lambda = c/f_c = 0.125$ m. The antenna configuration employs a one-dimensional linear fluid antenna with standard half-wavelength port spacing $\lambda/2$, utilizing maximum channel gain port selection strategy. The channel model parameters represent realistic wireless propagation conditions with average channel gains $\eta_A = 1.0$ and $\eta_E = 0.5$ for Alice advantage scenario, normalized noise variances $\sigma_A^2 = \sigma_E^2 = 1.0$, complex Gaussian channel distribution $g_k \sim \mathcal{CN}(0, \eta)$.
	
	For accurate evaluation of the derived analytical expressions, the numerical integration parameters include integration range $H = 8\sqrt{\eta_A}$, outer quadrature points $U_p = 30$, inner quadrature points $U_l = 20$, employing Gauss-Chebyshev quadrature method with convergence tolerance $10^{-6}$.
	
	The GS algorithm employs grid resolution $G = 30$ for standard scenarios and $G = 100$ for high-precision scenarios, with parameter ranges $N_A \in [5, 30]$ and $P_{tx} \in [0.1, 20]$ Watt, subject to the constraint that maximum total ports $N_A + N_E \leq 40$. The GD  algorithm utilizes learning rate $\alpha = 0.01$ for standard scenarios and $\alpha = 0.001$ for fine-tuning, with maximum iterations $T = 100$ for standard scenarios and $T = 1000$ for high-precision scenarios, numerical gradient step $\epsilon = 0.01$, and convergence criterion $|\nabla f| < 10^{-6}$. 
	
	Monte Carlo simulations are performed with number of realizations $N_\text{sim} = 5 \times 10^4$ for standard scenarios and $N_\text{sim} = 10^5$ for high-precision scenarios, target secrecy rate $R_s = 0.5$ bits/s/Hz for SOP analysis, 95\% confidence intervals, and fixed random seed for reproducibility.
	
	\subsection{VBCM Performance Validation and SOP Analysis}

	\begin{figure}[t]  \centering
		\includegraphics[width=3.5in]{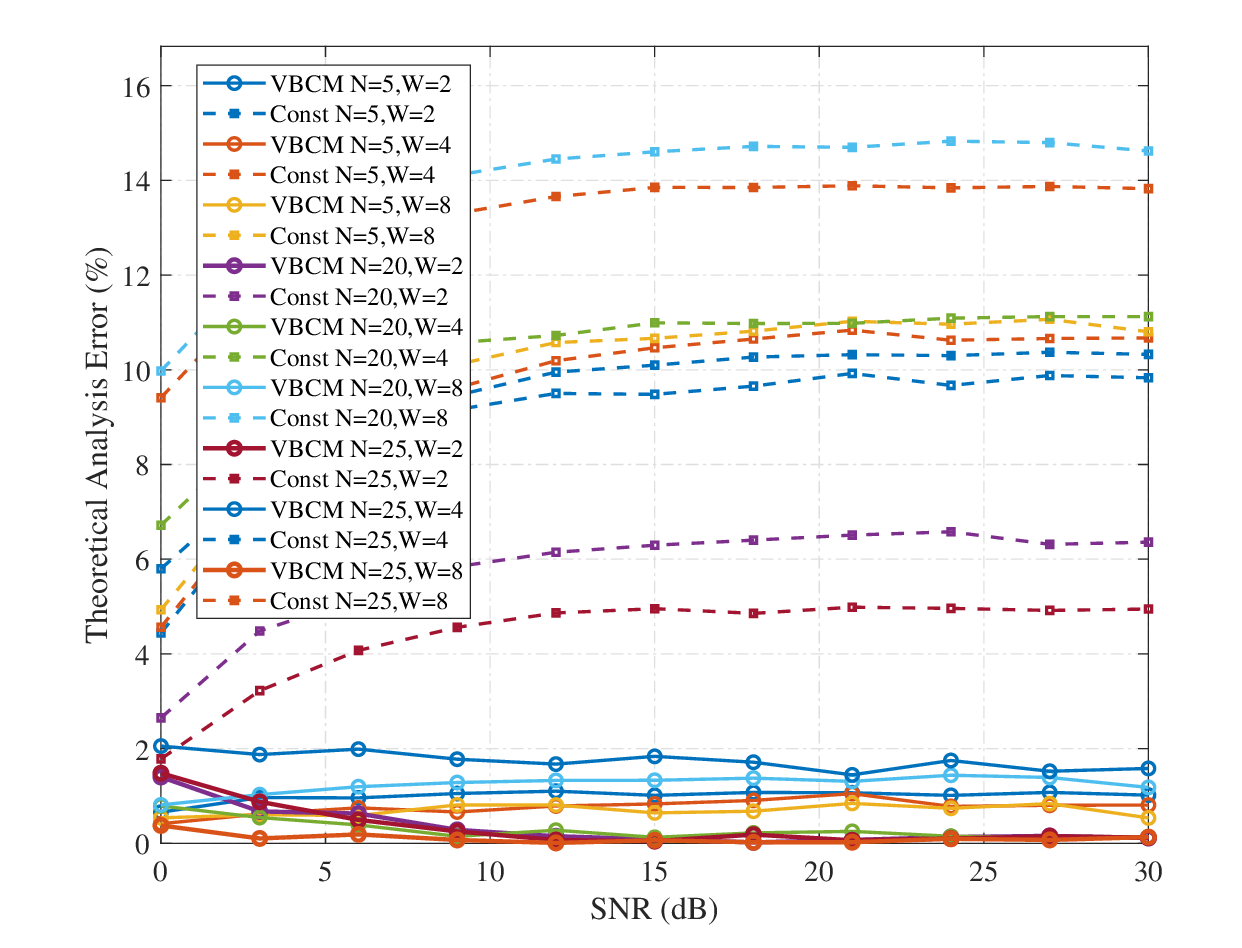}
		\caption { VBCM relative improvement analysis for multiple parameter combinations: SOP performance enhancement across different Port numbers ($N$) and normalization parameters ($W$).}\label{SOPv2}
	\end{figure}
	
	Fig.~\ref{SOPv1}  presents a comprehensive comparison between VBCM and the uniform constant correlation model \cite{FAS22} across different parameter combinations, varying both the number of ports $N_A\&N_E \in \{5, 20, 25\}$ and the normalization parameter $W \in \{2, 4, 8\}$. The results demonstrate that VBCM theoretical analysis achieves remarkable agreement with simulation results, with relative errors typically below 1\% across all configurations, in contrast to the constant correlation model with typical errors in the range $10-20$\%. \textit{This superior accuracy stems from VBCM's ability to capture the non-uniform spatial correlation structure inherent in FAS security deployments, as opposed to the uniform constant correlation model \cite{FAS22} which assumes all off-diagonal elements are identical.} The theoretical curves  closely track the simulation results, validating our analytical framework's reliability. Notably, the performance gap between VBCM and uniform constant correlation models becomes more pronounced at higher SNR values, where both approaches achieve lower SOP but VBCM maintains its advantage. \textit{This trend reflects VBCM's enhanced capability to model the complex correlation patterns that emerge in larger antenna arrays, where the spatial diversity benefits of FAS can be more accurately characterized}. Even for compact configurations ($N=5$), VBCM provides substantial accuracy improvements of 18-19\% over uniform constant correlation models, demonstrating its effectiveness across the full spectrum of practical FAS deployments.
	
	Fig.~\ref{SOPv2} illustrates the theoretical analysis error comparison between VBCM and uniform constant correlation models across different parameter combinations. The results reveal several insights: \textit{First, VBCM consistently achieves superior theoretical accuracy, with analysis errors remaining below 2\% for dense configurations ($N=25, W=8$) across the entire SNR range}. Second, for moderately dense configurations ($N=20$), VBCM maintains errors within 6-7\% while uniform constant correlation models exhibit significantly higher errors (10-15\%), particularly at high SNR values. \textit{Third, the accuracy advantage of VBCM becomes more pronounced as the antenna array size increases, demonstrating its particular effectiveness in analyzing large-scale FAS security systems}. Even for sparse (less dense) configurations ($N=5$), VBCM maintains theoretical errors below 10\%, while the uniform constant correlation model \cite{FAS22} show consistent degradation with increasing SNR, reaching errors of up to 14\%. This comprehensive error analysis validates the superior reliability of the VBCM-based analytical framework for FAS security evaluation across diverse deployment scenarios.
	

	\subsection{Correlation Matrix Structure Analysis}
	
	\begin{figure*}[t]    \hspace{-2.3cm}
         		\includegraphics[width=8.2in]{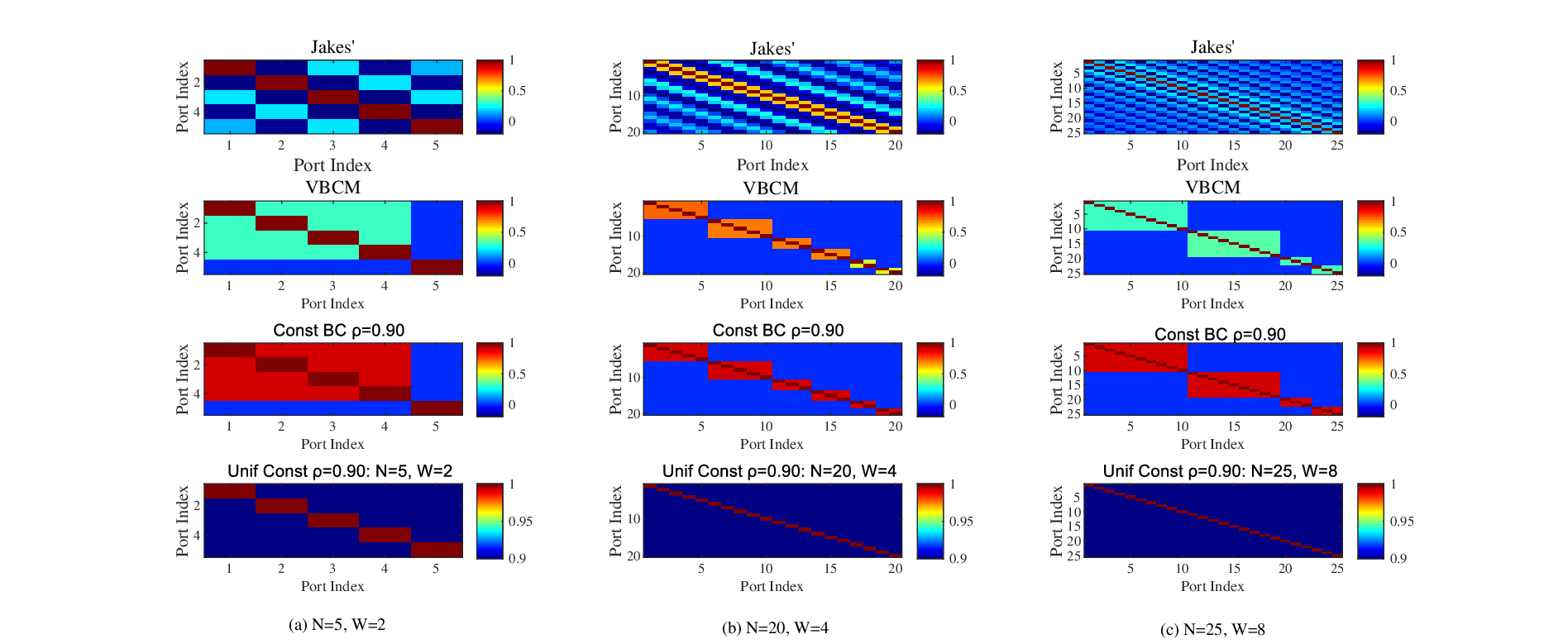}
		\caption {Correlation matrix structure comparison: VBCM vs. constant correlation models for different FAS configurations. The figure compares four correlation modeling approaches: (Row 1) Jakes' model (ground truth), (Row 2) VBCM approximation, (Row 3) Block-Correlation model with a single correlation parameter, and (Row 4) Uniform Constant Correlation model \cite{FAS22}, across three configurations: (a) N=5, W=2, (b) N=20, W=4, and (c) N=25, W=8.}\label{heap_mapv1}
	\end{figure*}
	Fig.~\ref{heap_mapv1} compares four correlation modeling approaches: (1) \textit{Jakes' model} (ground truth), (2) \textit{VBCM}, i.e., the block-correlation model from \cite{BC24} with different per-block correlation parameters, (3) \textit{Block-Correlation model \cite{BC24} with a single correlation parameter} ($\rho=0.90$) for all blocks, and (4) \textit{Uniform Constant Correlation model \cite{FAS22}} (all off-diagonal elements identical, ignoring spatial structure).
	
	\textit{The block-correlation models, either with single or multiple (VBCM) correlation parameters, provide relatively accurate representations of the original spatial correlation patterns (Jakes' model), while the constant correlation approach \cite{FAS22} yields an oversimplified uniform structure}. VBCM successfully identifies block structures with varying correlation intensities, from simple 2-block patterns ($N=5$) to complex multi-block structures ($N=25$). \textit{The block-correlation model with single parameter maintains block structure but simplifies correlation variations, while the uniform constant correlation model \cite{FAS22} fundamentally misrepresents the underlying spatial physics by reducing all correlations to a uniform matrix}. The VBCM modeling framework thus represent a good balance between accuracy and analytical tractability.
	\subsection{Comprehensive VBCM Performance Analysis}
	
	\begin{figure*}[t] \hspace{-1.3cm} \centering
		\includegraphics[width=7in]{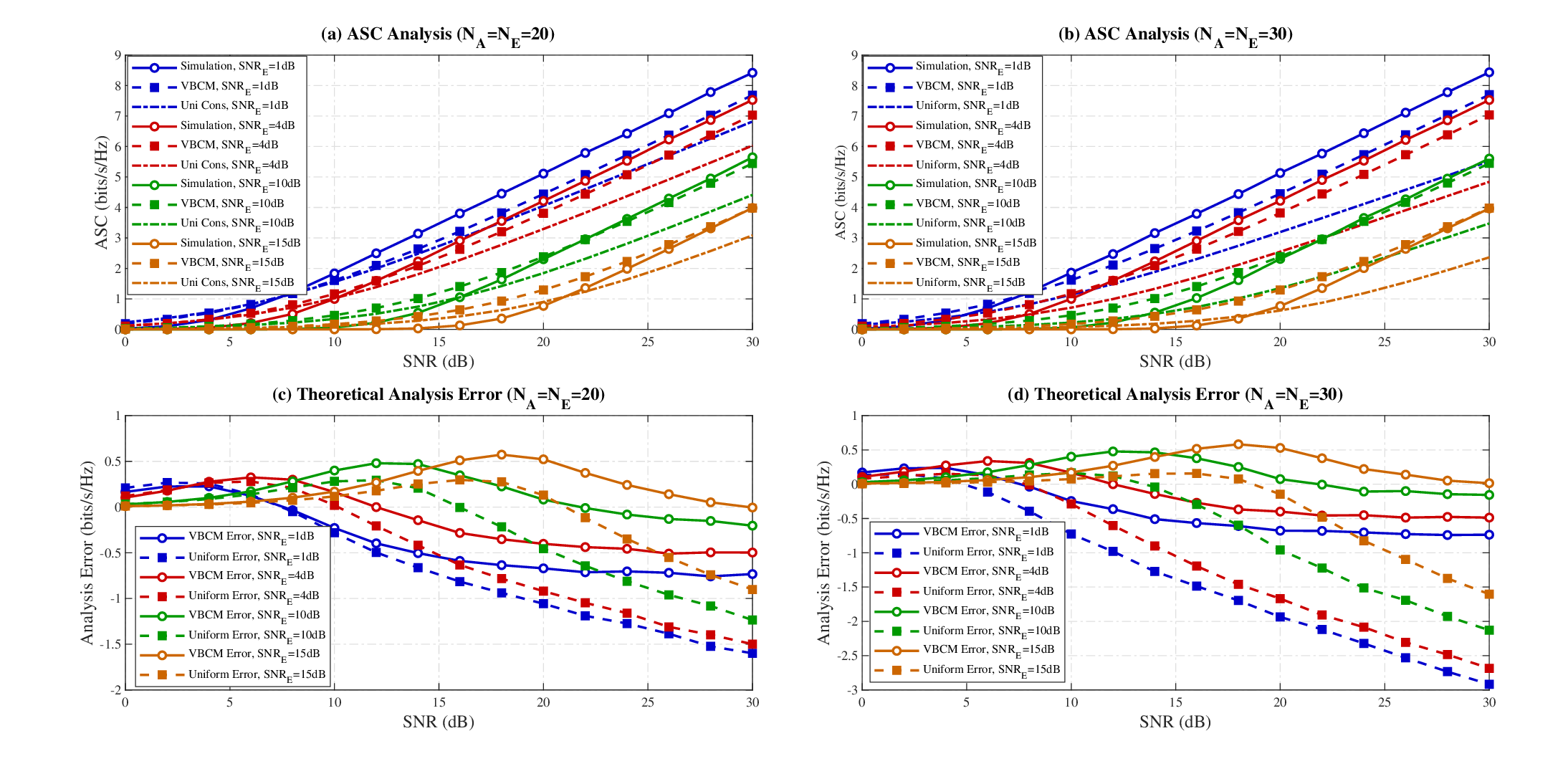}
		\caption {Comprehensive ASC Performance Analysis: Monte Carlo Simulation, VBCM Theory, and Uniform Constant Correlation Model under Varying System Configurations and Threat Scenarios. The figure compares three approaches: (1) Monte Carlo simulation based on true Jakes' model (ground truth), (2) VBCM theoretical analysis, and (3) Uniform Constant Correlation model across different port numbers and eavesdropper threat levels.}\label{ASCv1}
	\end{figure*}
	
	Fig.~\ref{ASCv1} compares three approaches across different port numbers ($N_A=N_E \in \{20, 30\}$) and threat levels ($\text{SNR}_E \in \{1, 4, 10, 15\}$ dB): Monte Carlo simulation (ground truth), VBCM theory, and uniform constant correlation model \cite{FAS22}.
	
	\textit{VBCM theoretical results closely match Monte Carlo simulation, while the uniform constant correlation model significantly underestimates ASC, particularly at higher SNR values}. For 30-port systems, VBCM reaches nearly 8 bits/s/Hz (matching simulation), while the uniform model shows substantial deviations. Security performance is highly sensitive to eavesdropper capability as $\text{SNR}_E$ increases from 1 dB to 15 dB, maximum ASC drops to 2 bits/s/Hz.
	
	Fig.~\ref{ASCv1}(c) and (d) quantify theoretical analysis errors. \textit{Errors under VBCM modeling remain below 0.2 bits/s/Hz, while uniform constant correlation modeling \cite{FAS22} yields errors above 0.5 bits/s/Hz at high SNR}. This validates that VBCM successfully captures essential spatial correlation physics, while oversimplified uniform assumptions lead to significant modeling inaccuracy for FAS security analysis.
	
	\subsection{Threat-Adaptive Security Optimization Analysis}
	Fig.~\ref{threat_ana} presents a comprehensive 3D surface analysis of FAS security optimization, demonstrating how the ASC landscape and the corresponding optimal parameters ($N_A$, $P_{tx}$) shift under varying eavesdropper threat levels ($N_E \in \{5, 15, 25\}$). The analysis is conducted using the GS algorithm with normalization parameter $W = 3$, for Alice's ports ranging from $N_A = 5$ to $30$, and transmit power from $P_{tx} = 0.1$ to $20$ W.
	
	The 3D surface visualizations reveal several critical optimization characteristics. The ASC exhibits a clear peak structure, with the optimal operating point (marked by a red circle) shifting in response to the threat level. Under a \textit{low threat} ($N_E = 5$), the system achieves a maximum ASC exceeding 0.5 bits/s/Hz, with an optimal configuration requiring a moderate number of ports ($N_A \approx 20$) and transmit power ($P_{tx} \approx 10$ W). As the threat escalates to a \textit{medium level} ($N_E = 15$), the peak ASC drops to approximately 0.4 bits/s/Hz, and the system must compensate by increasing resources, shifting the optimal point to a larger port count ($N_A \approx 25$) and higher power ($P_{tx} \approx 15$ W).  
	
	The surface topology provides deeper design insights. \textit{The optimization landscape becomes increasingly constrained as the threat level rises, with steeper gradients around the optimal region, indicating that system performance becomes highly sensitive to parameter selection in hostile environments}. The clear color gradient from blue (low ASC) to red (high ASC) delineates the feasible operating regions, showing that effective security requires a careful balance between spatial diversity (via $N_A$) and power allocation (via $P_{tx}$). \textit{This analysis powerfully demonstrates the fundamental security-threat trade-off in FAS systems and underscores that effective threat mitigation requires an adaptive system reconfiguration that intelligently scales both port deployment and power investment}. The shrinking high-performance region under escalating threats highlights the critical need for precise, threat-aware optimization to maintain reliable secure communication.
	\begin{figure*}[t]  \centering
		\includegraphics[width=7in]{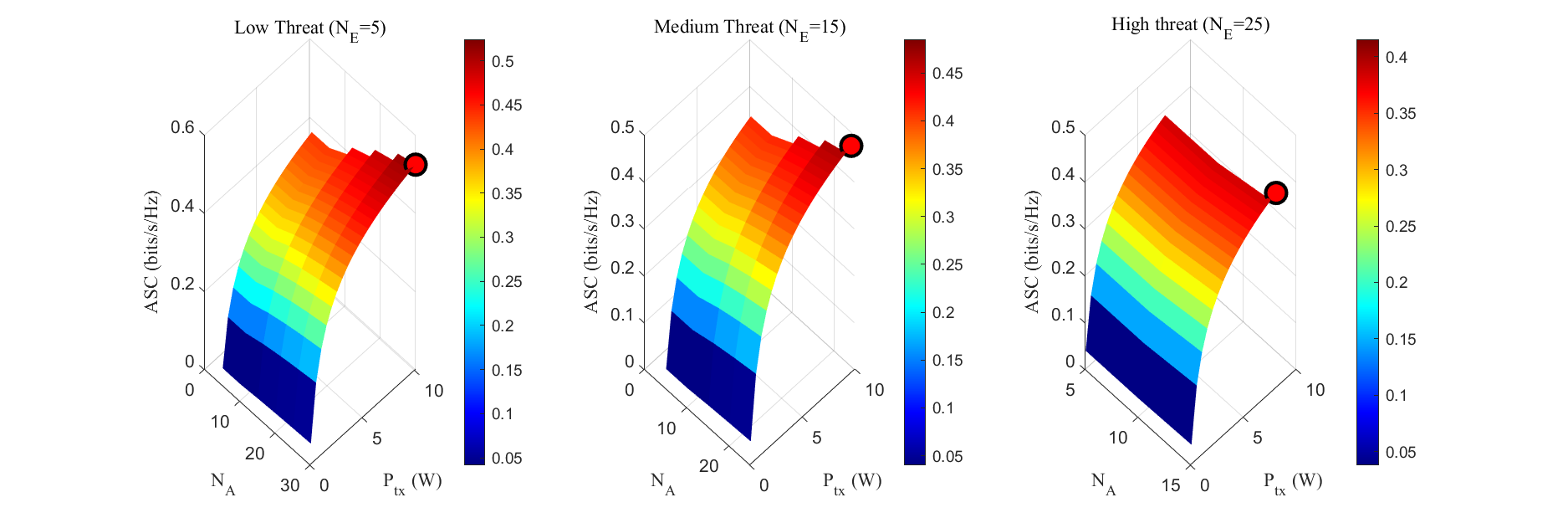}
		\caption{FAS Security Optimization Under Varying Threat Scenarios: 3D ASC Surface Analysis for Different Eavesdropper Capabilities ($N_E = 5, 15, 25$) Showing Optimal Parameter Selection.}\label{threat_ana}
	\end{figure*}

	\begin{figure}[t]  \centering
		\includegraphics[width=3.4in]{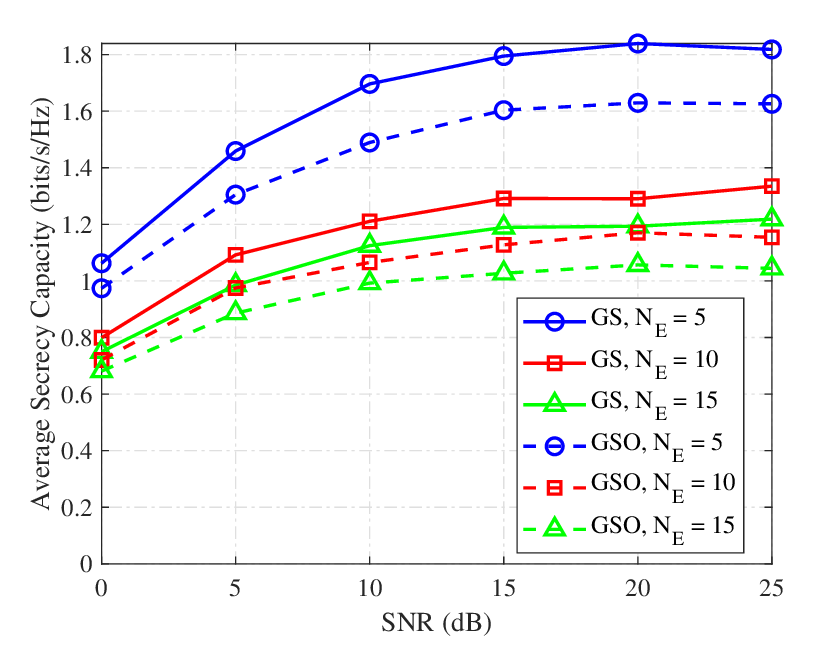}
		\caption{Algorithm Performance vs. SNR: ASC comparison of GS  and  GD under varying threat levels ($N_E \in \{5, 10, 15\}$) with a fixed Alice configuration ($N_A = 25$).}\label{algo1}
	\end{figure}
	
	\begin{figure}[t]  \centering
		\includegraphics[width=3.4in]{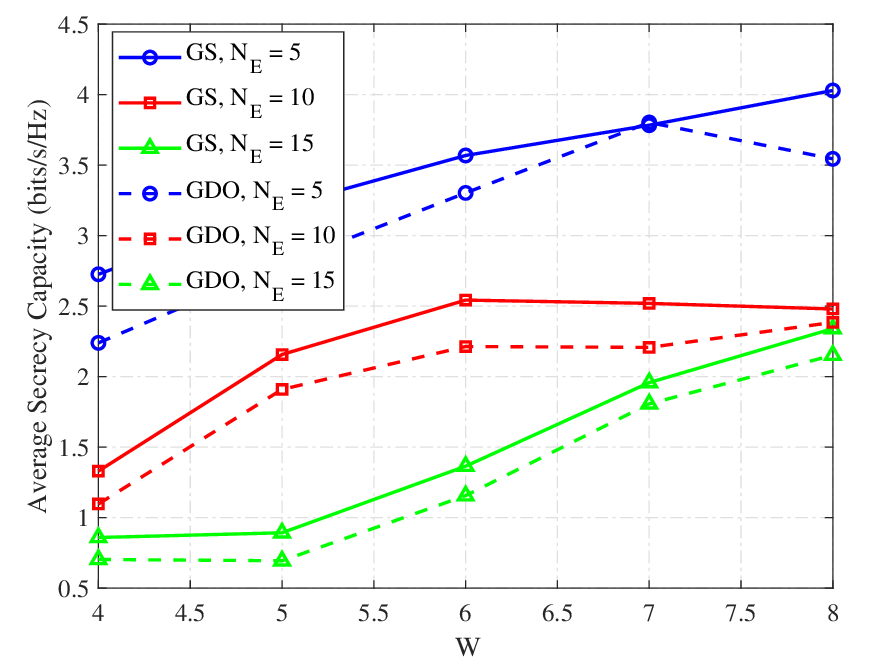}
		\caption{Parameter Sensitivity vs. W: ASC as a function of the antenna array normalization parameter ($W$) for both GS and GDO algorithms, demonstrating the threat-dependent optimal configuration.}\label{asc_w_algo1}
	\end{figure}

	\subsection{Algorithm Performance Comparison}
	\subsubsection{SNR-Based Performance Analysis}
	Fig.~\ref{algo1} presents a comprehensive performance comparison between the  GS  and GD algorithms. The analysis reveals that the GS algorithm (solid lines) consistently achieves higher ASC than the GD algorithm (dashed lines) across all tested SNR values and threat levels. Under a low threat condition ($N_E=5$), GS reaches a peak ASC of approximately 1.8 bits/s/Hz, outperforming GD by about 12.5\%. \textit{This consistent performance gap highlights the trade-off between the global optimality guaranteed by GS and the computational efficiency of GD, positioning GD as a viable, fast alternative for less critical applications.} As the threat level escalates to high ($N_E=15$), the maximum ASC achievable by GS drops to around 1.2 bits/s/Hz. \textit{Notably, the ASC curves begin to saturate at higher SNR values, indicating diminishing returns from simply increasing transmit power, especially under strong eavesdropping conditions.} This underscores the importance of optimizing other system parameters, such as port numbers and antenna geometry, to enhance security.
	
	\subsubsection{W-Parameter Impact Analysis}
	Fig.~\ref{asc_w_algo1} illustrates the system's sensitivity to the normalization parameter $W$, which dictates the antenna array's effective length. For the low threat scenario ($N_E = 5$), the ASC steadily increases with $W$, indicating that a larger antenna aperture is beneficial when the eavesdropper is weak. However, for medium ($N_E = 10$) and high ($N_E = 15$) threat levels, the optimal $W$ shifts to moderate values around 6-7 before performance slightly declines. \textit{This non-monotonic trend reveals a critical design trade-off: a larger $W$ enhances spatial diversity but also increases spatial correlation, which a stronger eavesdropper can potentially exploit.} \textit{The optimal $W$ therefore represents a threat-dependent balance between these competing effects, demonstrating that antenna array design must be adapted based on the anticipated security challenges.}
	
\subsection{Algorithm Efficiency and Selection Framework}

\begin{figure*}[t]
\centering
\includegraphics[width=7in]{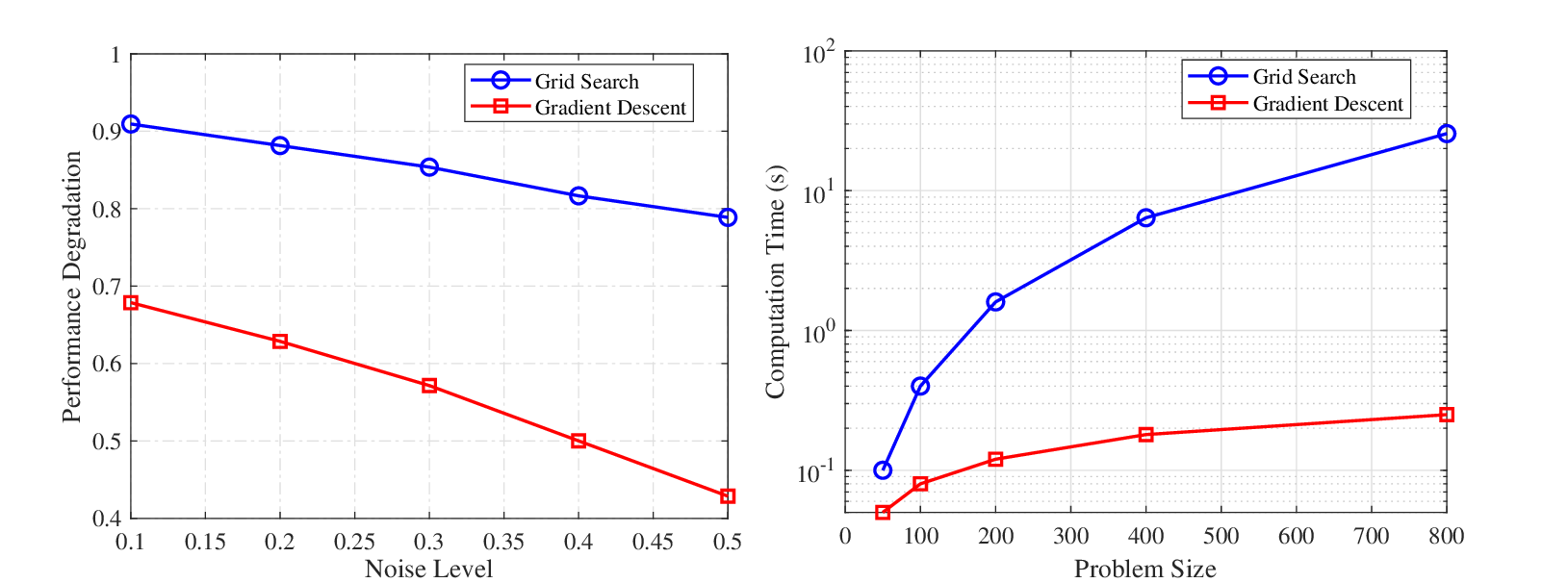}
\caption{Algorithm efficiency comparison: Computational cost and scalability of GS and GD.}\label{algo_ana}
\end{figure*}

To provide a complete optimization framework, we analyze the efficiency of both the GS and GD algorithms in Fig.~\ref{algo_ana}. The left figure illustrates the execution time versus the statistical uncertainty in the ASC evaluation, representing the ``noise" inherent in Monte Carlo-based channel simulations. The right figure shows the execution time as a function of the problem size (i.e., the number of parameters to optimize). The analysis clearly demonstrates the fundamental trade-off between the two algorithms. GD is consistently faster and more efficient, with its execution time remaining low and scaling linearly with problem size. In contrast, GS, while guaranteeing global optimality, exhibits significantly higher computational costs that grow exponentially, making it impractical for large-scale problems.

These distinct characteristics allow GS and GD to form a complete and practical optimization framework for FAS security. The choice of algorithm depends on the specific application requirements: \textit{GS is ideally suited for offline analysis, system design, and establishing performance benchmarks where global optimality is paramount and computational time is less of a constraint. Conversely, GD is the preferred choice for real-time applications and resource-constrained environments that demand rapid, adaptive optimization with near-optimal performance, especially in large-scale FAS deployments.} This dual-algorithm approach provides a versatile toolkit, enabling robust FAS security design across the entire lifecycle, from theoretical benchmarking to practical operational deployment.

\section{Conclusion}\label{sec:conclude}
This paper further developed and applied the VBCM framework (orginally proposed in \cite{BC24}) to comprehensive secrecy analysis in FAS. The key contributions include the development of a novel algorithm to suitably adjust the per-block correlation parameters, rigorous mathematical derivation of ASC and SOP expressions, and the design of two optimization algorithms for practical FAS security enhancement. The VBCM-based analysis demonstrates superior accuracy compared to constant correlation models, achieving relative errors below $5\%$ versus $10-15\%$. The framework proves particularly effective for compact FAS deployments with limited ports ($N<20$), where spatial correlation modeling is most challenging. The two optimization algorithms, GS and GD, provide versatile tools for different application scenarios, from offline system design to real-time adaptive optimization. Extensive numerical validation confirmed that FAS systems achieve substantial security improvements.

\end{document}